\documentclass[copyright,creativecommons]{eptcs}
\providecommand{\event}{TARK 2023} 

\usepackage{iftex}
\usepackage{amssymb}
\usepackage{amsthm}

\newcommand{\B}{{\mathcal B}}
\newcommand{\always}{{\Box}}
\newcommand{\commentout}[1]{}
\newcommand{\nciteyear}[1]{\cite{#1}}
\newcommand{\Pg}{{\sf Pg}}

\newcommand{\BibTeX}{\rm B\kern-.05em{\sc i\kern-.025em b}\kern-.08em\TeX}

\newtheorem{theorem}{Theorem}[section]
\newtheorem{example}{Example}
\newcommand{\fullv}[1]{#1}
\newcommand{\shortv}{\commentout}

\ifpdf
  \usepackage{underscore}         
  \usepackage[T1]{fontenc}        
\else
  \usepackage{breakurl}           
\fi

\title{Joint Behavior and Common Belief}
\author{Meir Friedenberg
\institute{Department of Computer Science \\ Cornell University}
\email{meir@cs.cornell.edu}
\and
Joseph Y. Halpern
\institute{Department of Computer Science \\ Cornell University}
\email{halpern@cs.cornell.edu}
}
\def\titlerunning{Joint Behavior and Common Belief}
\def\authorrunning{M. Friedenberg \& J.Y. Halpern}
\begin{document}
\maketitle

\begin{abstract}
	For over 25 years, common belief has been widely viewed as
	necessary for joint behavior.  But this is not
	quite correct.  We show by example that what can naturally be thought
	of as joint behavior can occur without common
	belief.  We then present two variants of common belief that can lead
	to joint behavior, even without standard common belief ever
	being achieved,
	and show that one of them, \emph{action-stamped} common belief, is in
	a sense necessary and sufficient for joint behavior.
	These observations are significant because, as is well known, common
	belief is quite difficult to achieve in practice, whereas these
	variants are more easily achievable.
\end{abstract}

\section{Introduction}


The past few years have seen an uptick of interest in studying
cooperative AI, that is, AI systems that are designed to be effective
at cooperating.  Indeed, a number of influential researchers recently
argued that ``[w]e need to build a science of cooperative
AI \dots\ progress towards socially valuable AI will be stunted unless
we put the problem of cooperation at the centre of our research'' 
~\cite{DBHHLG21}.

One type of cooperative behavior is
\emph{joint behavior}, that is, collaboration scenarios where the success of
the joint action is dependent on all agents doing their
parts; one agent deviating can cause the efforts of others to be
ineffective.
The notion of joint behavior has been studied (in much detail) under
various names such as ``acting
together'', ``teamwork'', ``collaborative plans'', and ``shared
plans'',
and highly influential models of it were developed 
(see, e.g., ~\cite{bratman92, CL91, GK96, GrS90, LVC90, Tuomela05}). Efforts
were also made to engineer some of these theories into
real-world joint planning systems~\cite{tambe97, RS97}.  
Examples of the types of scenarios these works considered include
drivers in a caravan, where if any agent deviates it might lead the
entire caravan to get derailed, and a company of military helicopters,
where deviation on the part of some agents can lead to the remaining
agents being stranded or put in unnecessarily high-risk scenarios. 

All the earlier work agrees on the importance
of beliefs for this type of cooperation.   
In particular, because each agent would do her part only if she
believed that all of the other agents would do their part as well,
there is a widespread claim that \emph{common belief} (often called \emph{mutual
	belief}) of how 
the agents would behave was necessary.   That is, not only did
everyone have to believe all of the agents would act as desired, but
everyone had to believe everyone believed it, and everyone had to
believe that everyone believed everyone believed it, etc.  This, they
argued, followed from the fact that everyone acts only if they believe
everyone else will. 
(See, e.g., \cite{bratman92, CL91, GK96, GrS90, LVC90,Tuomela05} for
examples of this claim.) 

As we show in this paper, this conclusion is not quite right.  We do
not need common belief for joint behavior; weaker variants suffice.
Indeed, we provide a variant of common belief that we call 
\emph{action-stamped} common belief that we show is, in a sense,
necessary and sufficient for joint behavior.  The key insight is that
agents do not have to act simultaneously for there to be joint behavior.
If agent 2 acts after agent 1, agent 1 does not have to believe, when
he acts, that agent 2 currently believes that all agents will carry
out their part of the joint behavior.  Indeed, at the point that agent
1 acts, agent 2 might not even be aware of the joint action.  It
suffices that agent 2 believes \emph{at the point that she carries out
	her part of the joint behavior} that all the other agents will
believe at the points where they are carrying out their parts of the
joint behavior \ldots\ that everyone will act as desired at the
appropriate time.
If actions must occur simultaneously, then common belief is necessary
\cite{FHMV};
the fact that we do not require simultaneous actions is what allows us
to consider weaker variants of common belief.

%

Why does this matter?  Common belief may be hard to obtain (see
\cite{FHMV}); it may be possible to obtain action-stamped common belief
in circumstances where common belief cannot be obtained.   Thus, if we
assume that we need common belief for joint behavior, we may end up mistakenly
giving up on cooperative behavior when it is in fact quite feasible.


\fullv{
The rest of the paper is organized as follows.  In the next section,
we provide the background for the formal (Kripke-structure based)
framework that we use throughout the paper.
In Section~\ref{sec:time-stamped}, we give our first example
showing that agents can have joint behavior without common belief, and
define a variant of common belief that we call \emph{time-stamped
	common belief} which enables it to happen.  In
Section~\ref{sec:action-stamped}, we give a modified version of the
example where time-stamped common belief does not suffice for joint behavior, but \emph{action-stamped common belief}, which is yet more general, does.
In general, the group of agents involved in a joint behavior need not be static; it may change over time.
For example, we would like to view the firefighters at the scene of a fire
as acting jointly, but this group might change over time as
additional firefighters arrive and some firefighters leave.
In Section~\ref{sec:dynamic}, we show how action-stamped (and
time-stamped) common belief can be extended to deal with
the group of agents changing over time.
In Section~\ref{sec:sig}, we go into more detail regarding the
significance of these results.  
In Section~\ref{sec:correctness}, we show that there is a sense in
which action-stamped common belief is necessary and sufficient for
joint behavior.  Finally, in Section~\ref{sec:conclusion}, we conclude. 
}

\section{Background}

To make our claims precise, we need to be able to talk formally
about beliefs and time.  To do so, we draw on standard ideas from
modal logics and the runs-and-systems framework of Fagin et
al. \nciteyear{FHMV}. 

Our models have the form $M=(R, \Phi, \pi, \B_1, \dots,
\B_n)$. 
$R$ is a \emph{system}, which, by definition, is a set of
\textit{runs}, each of which describes a way 
the system might develop over time.  Given a run $r \in R$ and a time $n
\in \mathbb{N}_{\geq 0}$ (for simplicity, we assume that time ranges
over the natural numbers), we call $(r,n)$ a \textit{point} in the
model; that is, it describes a point in time in one way the system
might develop.  $\Phi$ is the set of variables.  In general, we will
denote variables in $\Phi$ with uppercase letters (e.g., $P$) and
values of those variables with lowercase ones (e.g., $p$). $\pi$ is an
\textit{interpretation} that maps each point in the model and each
variable $P \in \Phi$ to a value, denoting the value of $P$ at that
point.
(Thus, the analogue of a primitive proposition for us is a formula of
the form $P=p$: variable $P$ takes on value $p$.)
	Finally, for each agent $i$, there is a \textit{binary
		relation} $\B_i$ over the points in the model.  Two points $(r_1,
	n_1)$ and $(r_2, n_2)$ are related by $\B_i$ (i.e.,
	$(r_1,n_1),(r_2,n_2)) \in \B_i$) if the two points are
	indistinguishable to agent $i$; that is, if, at the point $(r_1,n_1)$,
	agent $i$ cannot tell if the true point is $(r_1,n_1)$ or $(r_2,n_2)$.
	We assume throughout that the $\B_i$
	relations satisfy the standard properties of a belief relation:
	specifically, they are \emph{serial} (for all points $(r,n)$, there
	exists a point $(r',n')$ such that $((r,n),(r',n')) \in \B_i$),
	\emph{Euclidean} (if $((r_1,n_1),(r_2,n_2))$ and
	$((r_1,n_1),(r_3,n_3))$ are in $\B_i$, then so
	\sloppy{is $((r_2,n_2),(r_3,n_3))$), and transitive.}
	These assumptions ensure that the standard axioms for belief hold; see
	\cite{FHMV} for further discussion of these issues.
	
To talk about these models, we use the language 
	generated by the following context-free grammar: 
	\[
	\varphi := P=p \;|\; \neg \psi \;|\; \psi_1 \wedge \psi_2 \;|\; B_i\psi \;|\; E_{G}\psi \;|\; C_{G}\psi,
	\]
	where $P$ is a variable in $\Phi$, $p$ is a possible value of $P$, and
	$G$ is a non-empty subset of the 
	agents.  The intended reading of $B_i \psi$ is that agent $i$ believes
	$\psi$; for $E_{G}\psi$ it is that $\psi$ is believed by everyone in
	the group $G$; and for $C_{G}\psi$ it is that $\psi$ is common belief
	among the group $G$. 
	
	We can inductively give semantics to formulas in this language relative to points in the above models.  The propositional operators $\neg$ and $\wedge$ have the standard propositional semantics.  The other operators are given semantics as follows:
	\begin{itemize}
		\item $(M,r,n) \vDash P=p$ if $\pi((r,n),P) = p$, 
		\item $(M,r,n) \vDash B_i \psi$ if $(M,r',n') \vDash \psi$ for
		all points $(r', n')$ such that $((r,n), (r', n')) \in \B_i$, 
		\item $(M,r,n) \vDash E_G \psi$ if $(M,r,n) \vDash B_i \psi$ for all $i \in G$
		\item $(M,r,n) \vDash C_G \psi$ if $(M,r,n) \vDash E^{k}_G \psi$ for all $k \geq 1$, where $E^{1}_G \psi := E_G \psi$ and $E^{k+1}_G \psi := E_G(E^{k}_G \psi)$.
	\end{itemize}
	
	There are a number of axioms that are valid in these models.  Since they are
	not relevant for the points we want to make here, we refer the reader
	to \cite{FHMV} for a discussion of them.
	
	\section{Time-Stamped Common Belief}\label{sec:time-stamped}
	
	We now give our first example showing that
	joint behavior does not require common belief.
	We do not define joint behavior here; indeed, as we
	said, there are a number of competing definitions in the literature
	\cite{LVC90, CL91, GK96, GrS90}.  But we hope the reader will agree
	that, however we define 
	it, the example gives an instance of it.
	
	\begin{quote}
		General $Y$ and her forces are standing on the top of a hill.
		Below them in the valley, the enemy is encamped.  General $Y$
 knows that her forces are not strong enough to defeat the enemy on
their own.  
She also knows that General $Z$ and his
troops, though knowing nothing of the encamped enemy, will arrive on the hill the next day at noon on the way back from a training exercise. 
Unfortunately though, General $Y$ and her troops must move on before then.  
Thankfully, all generals are trained for how to
		deal with this situation.  Just as her training recommends,
		General $Y$ sets up traps that will delay
		the enemy's retreat, and leaves one soldier behind to
                inform General $Z$ of the traps upon his
		arrival.  At 11:30 the next morning, General $Y$ receives a
		(false) message informing her that General $Z$ and his troops
		have been captured, and thus (incorrectly) surmises that the
		enemy will live to fight another day.  What in fact happens
		is that General $Z$'s troops arrive at noon and attack
		the enemy, the enemy attempts to retreat and is stopped by
		General $Y$'s traps, and the enemy is successfully defeated. 
	\end{quote}
	
	Clearly, Generals $Y$ and $Z$ jointly defeated the enemy.  
	Yet they never achieved common belief
	of what they were doing.  Before noon,
	General $Z$ didn't even think that the enemy was there, and 
	from 11:30 on, General $Y$ thought that General $Z$ would never arrive.  It
	follows that there was no 
	point at which they could have had common belief.  
	So what is going on here? 
	
	What this example suggests is that there are times when a type of
	\textit{time-stamped common belief} (cf., \cite{FHMV,HM90}) suffices
	to enable joint 
	behavior.  Intuitively, on the first day, General $Y$ believed that at
	noon on the second day General $Z$ would act, attacking the enemy.
	Similarly, at noon on the second day, General $Z$ believed that General $Y$ had
	acted the day before, setting up the necessary traps.  They also
	hold higher-order beliefs; for example, at the time she set the traps,
	general $Y$ believed that at noon the next day general $Z$ would believe
	that she had set the traps, otherwise she wouldn't have wasted the
	resources to set them, and so on.
	Much as in the usual case of common
	belief, these nested beliefs extend to arbitrary depths.
	What sets this example apart from those considered by earlier work is
	that, whereas in the earlier work agents needed to believe others would act as
	desired \emph{at the same point}, here the agents need to believe only that
	others will act as desired \emph{at the points where they're supposed
	to act for the joint behavior}.  This suggests that time-stamped
common belief can suffice for joint behavior.


We can capture this type of time-stamped common belief formally
with the following additions to the logic and semantic models above.
Syntactically, we add two 
more operators to the language, $E^{t}_{G} \psi$ and $C^{t}_{G} \psi$,
where $G$ is a set of agents.  
We then add to the semantic model a function $t$ 
that maps each agent and run to a
non-negative integer.
The intended reading of these is ``each agent $i \in G$ believes at
the time $t(i,r)$ that $\psi$'' and ``it is
time-stamped-by-$t$ common belief among the agents in $G$ that
$\psi$'', respectively.  We give semantics to these operators as follows: 
\begin{itemize}
	\item $(M,r,n) \vDash E^{t}_{G} \psi$ if $(M,r,t(i,r)) \vDash
	B_{i} \psi$ for all $i \in G$ 
	\item $(M,r,n) \vDash C^{t}_{G} \psi$ if $(M,r,n) \vDash E^{t,k}_G \psi$ for all $k \geq 1$, where $E^{t,1}_G \psi := E^{t}_G \psi$ and $E^{t,k+1}_G \psi := E^{t}_G(E^{t,k}_G \psi)$.
\end{itemize}
These definitions are clearly very similar to the (standard)
definitions given above for $E_{G} \psi$ and $C_{G} \psi$,
except that the beliefs of each agent $i \in G$ in run $r$ is
considered at the time $t(i,r)$.
It follows from the semantic definitions that  $E^t_G \psi$ and $C^t_G
\psi$ hold at either all 
points in a run or none of them.  

In the example above, this notion of time-stamped common
belief \textit{is} achieved if we take $t(Y,r)$ to be the time in run
$r$ that $Y$ laid the traps (which may be different times in different
runs) and take $t(Z,r)$ to be the time that $Z$ arrived in run $r$
(which was noon in the actual run, but again, may be different times
in different runs), 
provided that it is (time-stamped) common belief that both $Y$ and $Z$
will follow
their training.
That is, when $Y$ lays the
traps, $Y$ must believe that $Z$ will believe when he arrives that $Y$
laid the
traps, $Z$ will believe when he arrives that $Y$
believed when she laid the traps that he
would believe when he arrived that $Y$ laid the traps, and so on.  
The key point here is that time-stamped common
belief can sometimes suffice for achieving cooperative behavior, even
without standard common knowledge. 

Our notion of time-stamped common belief is a
generalization of (and was inspired by) Halpern and Moses' notion of
(time-$T$)
\emph{time-stamped common knowledge}.
Roughly speaking, for them, time-$T$ time-stamped common knowledge
of $\phi$
holds among the agents 
in a group $G$ if every agent $i$ in $G$ knows $\phi$ at
time $T$ on her clock, all agents in $G$ know at time $T$ on
their clock that all agents in $G$ know $\phi$ at time $T$ on their clock,
and so on (where $T$ is a fixed, specific time).  If it is common knowledge
that clocks are synchronized, then time-stamped common knowledge
reduces to common knowledge.  If we take $t(i,r)$ to be the time in
run $r$ that $i$'s clock reads time $T$ (and assume that it is
commonly believed that each agent's clock reads time $T$ at some point
in every run), then their notion of time-stamped common knowledge
becomes a special 
case of our time-stamped common belief.  But note that with time-stamped
common belief, we have the flexibility of referring to different times
for different agents, and the time does not have to be a clock
reading; it can be, for example, the time that an event like laying
traps occurs.

\section{Action-Stamped Common Belief}\label{sec:action-stamped}

There is an even more general variant of common belief that can
suffice for joint behavior.  
What really mattered in the previous example is that everyone had the
requisite beliefs  
at the times that they were acting.  But there need not necessarily
be only one such point per agent per run; an agent might act multiple
times 
as part of the plan, as the following modified version of the
story illustrates:

\begin{quote}
	\commentout{
		General $Y$ and her forces are standing on the top of a hill.
		Below them in the valley, the enemy is encamped.  General $Y$
		knows that her forces are not single-handedly strong enough to
		defeat the enemy.  But she also knows that General $Z$ and his
		troops are expected to arrive on the hill on the opposite side
		at some time in the near future, although she and her troops
		must move on before then.  Thankfully, all generals are
		trained for how to deal with this situation.  Just as her
		training recommends, General $Y$ sets up traps 
		that will delay the enemy's retreat, and leaves one soldier
		behind to go to the opposite hill and inform General $Z$ of the
		traps upon his arrival.  The next morning, General $Y$ receives
		a (false) message informing her that General $Z$ and his troops
		have been captured, and thus (incorrectly) surmises that the
		enemy will live to fight another day.  What in fact happens,
		though, is that General $Z$'s troops arrive later that day and
		are informed by the remaining soldier that \emph{at some point
			not too long ago} General $Y$'s troops set traps. They attack
		the enemy, the enemy attempts to retreat and is stopped by
		General $Y$'s traps, and the enemy is successfully defeated. 
	} 
	
	General $Y$ and her forces arrive to the south of the town where 
	the enemy forces are encamped.  General $Y$
        	knows that her forces are not strong enough to
	defeat the enemy on their own.  She also knows that General $Z$ and his
	troops are expected to arrive to the north of the city
	some time in the near future, though she and her troops
	must move on before then. The swiftly-coursing river
	prevents the enemy from escaping to the east.  But
	unfortunately, they can 
	still escape inland to the west. 
	Thankfully, all generals are
	trained for how to deal with this situation as well.  Just as her
	training recommends, General $Y$ sets up traps 
	that will delay the enemy's southward retreat and then, as she
	heads inland, 
	also sets up traps to the west, finally leaving one soldier
	behind to go north and inform General $Z$ of the
	traps upon his arrival.  The next morning, General $Y$ receives
	a (false) message informing her that General $Z$ and his troops
	have been captured, and thus (incorrectly) surmises that the
	enemy will live to fight another day.  What in fact happens
	is that General $Z$'s troops arrive later that day and
	are informed by the remaining soldier that, not
	too long ago, General $Y$'s troops set traps to the south 
	and west. They attack the enemy, the enemy attempts to retreat
	and is stopped by General $Y$'s traps, and the enemy is successfully defeated. 
\end{quote}

Again, Generals $Y$ and $Z$ jointly and collaboratively 
defeated the enemy,  
but time-stamped common belief doesn't suffice for this version of the story,
because we cannot specify a single time for General $Y$'s actions.
%
	Instead, what really matters is that when they are acting as 
part of a joint plan they hold the requisite (common)
beliefs. The joint plan need not be known upfront;
General $Z$ does not know what he will need to do to achieve the
common goal until he arrives at the scene.  To capture this new
requirement, we define a notion of \textit{action-stamped common
belief}.  

We begin by adding a special Boolean variable $ACTING_{i, G}$ for
any group $G$ and agent $i \in G$.  This variable is true (i.e., takes value 1, as
opposed to 0) at a point $(r,n)$ if the agent
$i$ is acting towards the group plan of $G$ at $(r,n)$ and false otherwise.
So for the generals, $ACTING_{Y, G} = 1$ would be true when she lays the
traps, $ACTING_{Z, G} = 1$ would be true at the point when he attacks,
and they'd both be false otherwise (where $G = \{ Y, Z \}$).
We often write $ACTING_{i,G}$ and $\neg ACTING_{i,G}$ instead of $ACTING_{i,G} = 1$ and $ACTING_{i,G} = 0$, and similarly for other Boolean variables.
By using $ACTING_{i,G}$, we can abstract away from what actions are
performed; we just care that some action is performed by agent $i$
towards the group plan, without worrying about what that action is.

As in the case of time-stamped common belief, we add two modal operators to the
language (in addition to the variables $ACTING_{i,G}$).  Let $G$ be
a set of agents.  $E^{\mathbf{a}}_{G} \psi$ 
then expresses that, for each agent $i \in G$, 
whenever $ACTING_{i,G}$ holds (it may hold several
times in a run, or never), $i$
believes $\psi$. $C^{\mathbf{a}}_{G} \psi$ then
defines the corresponding notion of common belief for the points at
which agents act at part of the group. 

We give semantics to these modal operators as follows:
\begin{itemize}
\item $(M,r,n) \vDash E^{\mathbf{a}}_{G} \psi$ if
for all $n'$ and all $i \in G$ such that $(M,r,n') \vDash
ACTING_{i,G} = 1$, it is also the case that $(M,r,n') \vDash
B_{i} \psi$.
\item $(M,r,n) \vDash C^{\mathbf{a}}_{G} \psi$ if $(M,r,n)
\vDash E^{\mathbf{a},k}_G \psi$ for all $k \geq 1$, where
$E^{\mathbf{a},1}_G \psi := E^{\mathbf{a}}_G \psi$ and
$E^{\mathbf{a},k+1}_G \psi :=
E^{\mathbf{a}}_G(E^{\mathbf{a},k}_G \psi)$. 
\end{itemize}

\commentout{
Formally, we give semantics to each of these as follows:
\begin{itemize}
\item $(M,r,n) \vDash E^{\vec{\mathbf{a}}}_{G} \psi$ if for
all $1 \leq i \leq len(G)$ there exists an $n' \geq 0$ such
that $(M,r,n') \vDash A_{G_i} = a'$ for some $a' \in
\vec{\mathbf{a}_i}$ and $(M,r,n') \vDash B_{i} \psi$. 
\item $(M,r,n) \vDash C^{\vec{\mathbf{a}}}_{G} \psi$ if $(M,r,n) \vDash E^{\vec{\mathbf{a}},k}_G \psi$ for all $k \geq 1$, where $E^{\vec{\mathbf{a}},1}_G \psi := E^{\vec{\mathbf{a}}}_G \psi$ and $E^{\vec{\mathbf{a}},k+1}_G \psi := E^{\vec{\mathbf{a}}}_G(E^{\vec{\mathbf{a}},k}_G \psi)$.  (\textbf{Actually, I'm not quite sure if this is the right way to define it.  It seems at least plausible that this leaves open the possibility, that different levels of the Es could in fact be at different points where actions in those sets are taken -- i.e.\ the base level beliefs might hold while taking one action while the higher level beliefs might hold while taking a different action.  I can't quite tell for sure whether we want to allow that, but my inclination is probably not.  If that's the case then we can avoid the problem by defining C in terms of if there exist times where those actions are taken such that \emph{time-stamped} common belief holds.})
\end{itemize}
}

Returning to the example, although the agents do not have time-stamped common
belief at all the points when they act, they do have action-stamped
common belief.
General $Z$ acted believing that General $Y$ had acted as expected,
and also believing that General $Y$ acted believing that he would act as
expected, and so on.

It is easy to see that time-stamped common belief can be viewed as a
special case of action-stamped common belief: Given a time-stamping
function $t$, we simply take 
$ACTING_{i,G}$ to be true at those points $(r,n)$ such that
$t(i,r)$ holds.

\fullv{It is worth noting that, in} \shortv{In} both this and the
previous section, the 
agents having a protocol in advance for how to deal with the situation
is not really necessary for them to succeed.  In the examples,
consider a scenario where generals are in fact not trained for how to
handle the situation, but instead General $Y$ has the brilliant idea
to lay traps and send a messenger to meet General
$Z$ upon arrival. As long as message delivery is reliable,
action-stamped common belief can be achieved and they can successfully
defeat the enemy. 

\section{Joint Behaviors Among Changing Groups}\label{sec:dynamic}

In practice, 
the members of groups change
over time.  For example, a group of firefighters may work together to
safely clear a burning building, but (thankfully!) they don't need to
wait until all the firefighters are on the scene, or even until it is known
which firefighters are coming, in order for the first
firefighters to begin. Instead, structures and guidelines allow the
set of firefighters who are on the scene to act cooperatively, even
without each firefighter knowing who else will show up.

The formalisms of the two previous sections assumed a fixed group $G$,
so cannot capture this kind of scenario.  But the changes necessary to
do so are not 
complicated.  Rather than considering (some variant of) common belief 
with respect to a fixed set $G$ of agents, we consider it with respect to an
\emph{indexical} set $S$, one whose interpretation depends on the
point.  More precisely, an indexical set $S$ is a function from points
to sets of agents; intuitively, $S(r,n)$ denotes the members of the
indexical group $S$ at the point $(r,n)$.  We assume that a model is
extended so as to provide the interpretation of $S$ as a function.

Our semantics for action-stamped common belief 
with indexical sets are now a straightforward generalization of the
semantics for rigid (non-indexical) sets: 
\begin{itemize}
\item $(M,r,n) \vDash E^{\mathbf{a}}_{S} \psi$ if
for all $n'$ and all $i \in S(r,n')$ such that
$(M,r,n') \vDash ACTING_{i,S}$, it is also the case
that $(M,r,n') \vDash B_{i} \psi$. 
\item $(M,r,n) \vDash C^{\mathbf{a}}_{S} \psi$ if $(M,r,n)
\vDash E^{\mathbf{a},k}_S \psi$ for all $k \geq 1$, where
$E^{\mathbf{a},1}_S \psi := E^{\mathbf{a}}_S \psi$ and
$E^{\mathbf{a},k+1}_S \psi :=
E^{\mathbf{a}}_G(E^{\mathbf{a},k}_S \psi)$. 
\end{itemize}
The only change here is that in the semantics of $E^{\mathbf{a}}_{S}$, 
we need to check the agents in $S(r,n)$ for each point.

Of course, we can also allow indexical sets in time-stamped common
belief in essentially the same way.  
Whereas
in the semantics of $E_G^t\psi$, we required that
$(M,r,n) \vDash E_G^t\psi$ if, for all $i \in G$, $(M,r,t(i,r)) \vDash
B_{i} \psi$, now we require that
$(M,r,n) \vDash E_S^t\psi$ if, for all agents $i$, if $i \in
S(r,t(i,r))$, then $(M,r,t(i,r)) \vDash B_{i} \psi$.  We care about what agent
$i$  believes at $(M,r,t(i,r))$ only if $i$ is actually in group $S$
at the point $(r,t(i,r))$.

\commentout{
While we think of the action-stamped version of common belief as more
interesting and important for cooperation, it is worth noting that we
can make a similar extension of time-stamped common belief to allow
for dynamic groups.  To start though, we extend our time-stamped
definition in a different manner, allowing $t$ to be a function of
both runs and agents.  Letting $t$ depending on the run allows for
agents to have uncertainty regarding the plan, and brings it closer to
what can be captured in the action-stamped setting, where an agent
acting as part of the cooperative plan at time $n$ in one run may not
be doing so in some other run.  (We could of course have allowed $t$
to depend on runs from the start in Section~\ref{sec:time-stamped},
but chose to start with the simplified version for ease of
understanding.) 

In the dynamic setting, the intuition for $t$ is slightly different,
saying only that agent $i$ will act with the requisite beliefs at
point $(r, t(r, G_i))$ \emph{if} she is in the group $G$ at $(r, t(r,
G_i))$.  Formally the semantics then become the following: 
\begin{itemize}
\item $(M,r,n) \vDash E^{t}_{G} \psi$ if for all $n'$ and all $G_i \in \pi((r,n'),G)$, if $t(r, G_i) = n'$ then $(M,r,n') \vDash B_{G_i} \psi$.
\item $(M,r,n) \vDash C^{t}_{G} \psi$ if $(M,r,n) \vDash E^{t,k}_G \psi$ for all $k \geq 1$, where $E^{t,1}_G \psi := E^{t}_G \psi$ and $E^{t,k+1}_G \psi := E^{t}_G(E^{t,k}_G \psi)$.
\end{itemize}
In essence, for $E^t_{G} \psi$, at every point all agent's who are in the group \emph{and} whose time it is must believe $\psi$.  Note that the semantics laid out here look much more similar in structure to those we gave for action-stamped common belief than our original definition of time-stamped common belief.  Of course, it would have been possible to give a definition for fixed groups that uses this structure and is equivalent to our original definition, but we chose to give that definition above to ease the reader in.

Though we haven't emphasized it until now, it is worth noting that all of our definitions specify \emph{run properties}; that is, for any of the above notions, they hold true at a point in a run if and only if they hold at all points in the run.  This is intentional.  Time-stamped and action-stamped common belief are both properties of whole runs, specifying that agents hold certain beliefs at all the times when they're supposed to act / points when they do act, rather than properties of particular points.
}

\commentout{
$t$ then must also be a variable, as the set being assigned to times will change between points. We can now make the following changes to account for changing groups:

For ``everyone believes'' we get
\begin{itemize}
\item $(M,r,n) \vDash E^{t}_{G} \psi$ if, for all $G_i \in \pi((r,n),G)$, it is the case that $G_i \in \pi((r,t^{n}_{G_i}),G)$ and $(M,r,t^{n}_{G_i})  \vDash B_{G_i} \psi$ , where $t^{n}_{G_i}$ is $\pi((r,n),t)(G_i)$.
\end{itemize}
What this says is that for everyone in the group at time $n$, they are each still in the group at the time that was assigned to them at time n (that is, $t^{n}_{G_i}$), and they believe $\psi$ at that time.  In essence, $E^{t}_{G}$ is a statement about the members of the group \textit{at time $n$} and their group-membership and beliefs at the times designated for each of them.

We also introduce a simple temporal operator to the language, $\always$.
\begin{itemize}
\item $(M,r,n) \vDash \always \psi$ if $(M,r,n') \vDash \psi$ for all $n'$.
\end{itemize}
(\textbf{Is there a standard symbol for this...? It seems like something that must be standard, but I'm not sure I've seen it.})  Note that this is not quite the standard $\Box$ operator; $\Box$ is usually that $\psi$ is true for all times \emph{from $n$ and on}, whereas we're using $\always$ here to mean it is true for all times, including those before $n$.
$\always E^{t}_{G} \psi$, for example, then means something along the lines of ``for all times $n$, the group members at $(r,n)$ each believe $\psi$ at the time they're assigned at $(r,n)$.''

Finally, time stamped common belief then becomes
\begin{itemize}
\item $(M,r,n) \vDash C^{t}_{G} \psi$ if $(M,r,n) \vDash (\always E^{t}_G)^k \psi$ for all $k \geq 1$, where $(\always E^{t}_G)^1 \psi := \always E^{t}_G \psi$ and $(\always E^{t}_G)^{k+1} \psi := \always E^{t}_G((\always E^{t}_G)^k \psi)$. 
\end{itemize}

We note two important aspects of these definitions.
The first is that $\always$-formulas express \emph{run} properties rather than \emph{point} properties; that is, if $(M,r,n) \vDash \always \varphi$ then $\always \varphi$ must in fact be true at \emph{all} points in the run $r$.

Though we didn't emphasize it before, our definitions of time-stamped common belief in Section~\ref{sec:time-stamped} and action-stamped common belief in Section~\ref{sec:action-stamped} are in fact run properties; because $G$ and $t$ are fixed, if $C^t_{G}$ was satisfied at any point in a run then it was satisfied at every point in that run.  

For our definitions in this section, in the case of dynamic groups, had we defined $C^t_G$ in the obvious way simply by nesting $E^t_G$ it would in fact \emph{not} have been a run property.  $E^t_G$ at point $(r,n)$ is always a statement about the beliefs of the particular agents who are in the group at $(r,n)$, but intuitively we don't want the current group members to have this ``special role'' in the determination of time-stamped common belief -- time-stamped common belief \emph{ought} to be a run property, dependent on the beliefs of the different group members over time, not just the beliefs of those in the group at $(r,n)$ \emph{about} the beliefs of others.

There are at least three natural-seeming ways to go from our notion of ``time-stamped everyone believes'' to a notion of time-stamped common belief that is a run property.  The first is what we have specified above, by repeated nesting of $\always E^t_G$.  The second is by prepending $\always$ to repeated nesting of $E^t_G$, giving semantics in terms of $\always (E^{t,k}_G)$ for all $k \geq 1$.  The third option would be to define time-stamped common belief and everyone believes relative to a \emph{fixed} time-stamp $t$ that maps a set of agents $G^*$ to times, such that $G^*_i \in G$ at $(r, t(G^*_i))$ for each $G^*_i \in G^*$ and all runs $r$.  All three of these options seem natural, and we saw no definitive reason to use any one of them.  That said, we decided not to use the third option because it seemed to preclude the type of ``dynamic plans'' scenarios that dynamic groups intuitively ought to help us be able to model.  Between the first two options, we decided to use the first one because it seemed to more directly parallel the earlier definitions of common belief and its variants, defining it solely in terms of repeated nesting of the exact same operator (here $\always E^t_G$).

The second important thing to note about these definitions is that getting the modeling right, and in particular specifying the right $G$, is very important.  Let's return to the example we started this section with, of firefighters showing up to a scene over time.  If you specified $G$ to be the set ``firefighters at that building'' then there very well might not be time-stamped common belief, because there might be a firefighter who had been at the scene two years earlier to buy a coffee and who at the time had no beliefs about any future fires that might occur there.  On the other hand, if you specified $G$ to be ``the set of firefighters present at the scene to fight that fire'', time-stamped common belief might well be achieved.  This actually seems correct to us --- the set of all firefighters at the building during the course of the run did \emph{not} have time-stamped cooperative behavior, but those present to fight that fire did.  Our point is only that the modeling choices are important, as they often are in these types of frameworks.

Lastly, we can of course define a notion of action-stamped common belief for dynamic groups, much the same way we have done here for time-stamped common belief.
} 

\section{Significance}\label{sec:sig}

In Sections~\ref{sec:time-stamped}-\ref{sec:dynamic} we showed that
action-stamped common belief can suffice to enable joint
behavior, whereas the prior work on the topic had assumed common
belief was necessary.  Why does this matter?  We argue that
it is important for two reasons:
1) misunderstanding the type
of belief necessary can lead to mis-evaluation of cooperative
capabilities, and 2) requiring common belief can unnecessarily
make cooperation impossible in scenarios where it is in fact possible
and could be quite beneficial. 


As part of the recent push for more research on cooperative AI, some
have argued that we should  
``construct more unified theory and vocabulary related to problems of
cooperation'' 
\cite{Dafoe20}.  
One important step in this program is (in our opinion)
formalizing the requirements for various types of cooperation,
including joint behavior.   This, in turn, 
requires understanding the level and type of (common) belief needed for
joint behavior.  As our examples have shown, full-blown common belief is
not necessary; weaker variants that are often easier to achieve
can suffice.
Relatedly, there has been a push to develop methods for
\emph{evaluating} the cooperative capabilities of agents, as a way of
developing targets and guideposts for the community~\cite{CAIFcall}.
Again, this will require understanding (among other things) what type
of beliefs are necessary for cooperation.  
Incorrect assumptions about the types of beliefs
necessary can lead to incorrect conclusions about the feasibility of
cooperation.  For 
example, if an evaluation system takes as given the assumption that it
is impossible for agents that cannot achieve common belief to behave
cooperatively, it may in fact lead to effective cooperative agents
being scored badly, leading to misdirected research. 


A second reason that it is important to clarify the types of beliefs
necessary for joint behavior is that misunderstanding them can lead to
systems unnecessarily aborting important cooperative tasks.  As
is well known, achieving true common knowledge can be
remarkably difficult in real-world systems, often requiring either a
communication system that guarantees truly synchronous delivery or
guaranteed bounded delivery time together with truly synchronized
clocks \cite{FHMV}.
Action-stamped common belief can sometimes be achieved when common
belief cannot.  To demonstrate the importance of this, we consider an
example from the domain of urban search and rescue, a domain where 1)
the use of multi-agent systems consisting of humans and AI agents has
long been considered and advocated for, 2) the types of teamwork
necessary can be complex, and 3) there is some evidence of potential
adoption, having been used, for example, at a small scale in the
aftermath of September 11th~\cite{CM03,KTNMTSS99,KT01,LN13,QTP20,SYDY04}.
Though the example we give is a simple, stylized case, the domain is sufficiently complex that we would expect these types of issues to arise in practice if systems were deployed at scale.

\begin{example}
An earthquake occurs, causing a large building to collapse.
The nearest search and rescue team arrives on scene, and the
incident commander has to decide how to proceed.
%
The team has determined that the
structure is stable and will not collapse, and so is safe to
enter. However, attempting to exit the building may disrupt
the structure and cause harm. 
The incident commander determines that
there are two reasonable options:
\begin{enumerate}
  \item Wait a week for a heavy piece of machinery that will certainly
be able to safely lift the roof of the collapsed building on
its own and allow rescuers safe access to the building.

\item A team can enter the building and restabilize parts of the roof.
The restabilization would not be enough to make it safe to
exit\shortv{, }\fullv{---in
fact, it would require adjusting the structure in ways that would make
an attempt to exit even more risky---}but it would allow a more easily
accessible robotic system to safely remove the
roof piece by piece, allowing the rescuers and anyone trapped
inside to safely escape.
\end{enumerate}

The
incident commander decides it is best not to wait, and so takes
the second, joint-behavior-based approach.  He sends the team of
rescuers in to 
begin the necessary process, and tells them 
the full plan and that he expects it
will be 2-3 hours before the robot arrives on scene. 
%
The group enters the wreckage and secures it in the necessary ways,
as planned.  But it turns out that the earthquake
affected many buildings, so the robot is in high demand.  It
ends up taking close to 8 hours for the robot to arrive on
scene.
When the robot arrives, the incident commander enters the relevant
information in the robot's system---namely, the full plan 
and that the restabilization has been carried out---so the robot can
carry out its part of the specified cooperative plan. 

If the robot's model of joint behavior requires
common belief, a problem will arise.  At no point is there
ever common belief of the joint behavior.  Before the robot
arrives, the robot certainly has no belief about the joint
behavior.  And 
when the robot arrives, it must consider the
possibility that, because of the delay, the rescuers have given
up hope of the robot arriving and concluded that they may have
to wait a full week until the larger piece of machinery is
available.  Even if this isn't actually the case, because the robot
considers it possible,
common belief will not be achieved.  So
if the robot assumes that joint behavior requires common belief, it
will determine that the joint behavior cannot be
carried out. Thus, everyone will have to wait a week for
the heavier machinery, risking the lives of anyone trapped
inside.

If, on the other hand, the robot's theory of joint behavior is
based on action-stamped common belief, the task will be 
properly and safely carried out as soon as the robot arrives:  When the rescuers
perform their part, they believe that the robot will arrive soon
and perform its part of the task.  Similarly, when the robot arrives
and the incident commander enters the relevant information, the robot
believes that the rescuers held those beliefs when acting
(and therefore performed the required adjustments).
The
rescuers believed that the robot would hold these beliefs when
it arrived, the robot believed they would, and so on.  The
fact that the robot 
arrived later than expected and the rescuers may have started
to have uncertainty about the plan doesn't affect the
requisite beliefs because all that matters are the beliefs of
the agents at the points where they act.  
\end{example}

This example highlights the value of getting the types of beliefs
necessary right; getting the theory right, and basing it on
action-stamped common belief instead of standard common-belief, can
enable cooperation in a range of important scenarios where standard
common belief is impossible or difficult to achieve, whereas
action-stamped common belief may be easily attainable. 

\commentout{
\section{Shared Plans and Joint Intentions}

So if standard common belief isn't strictly necessary for achieving cooperative behavior, why was it used in all the previous work?  The answer seem to us to be, simply, that they weren't trying to \emph{define} cooperative behavior.  Instead, they were trying to \emph{engineer} collaborative or cooperative behavior.  And to do that, they tried to come up with \emph{mental states} that could be used to ensure agents would behavior cooperatively, with one group for example focusing on the notion of joint intentions as a way to ensure cooperative behavior and another developing a framework of SharedPlans.  And indeed, had we been formalizing these notions we too would likely have formalized them using common belief.  For both of them, after all, it seems natural to ask ``when did it happen?''  ``When did the agents have a joint intent?''  ``When did the agents have a shared plan?''  And in order to have a specific point in time when it occurred, neither of the notions we developed in this paper would suffice.  However, when addressing the question ``when did the agents behave cooperatively?'' the answer may well be ``this one did his part at time $t$ and this one did hers at time $t'$''.  And similarly, it needn't inherently be the case that agents collectively have some sort of atomic mental state that is necessary for cooperation and occurs at a single time.  

From an engineering perspective, the prior work may well have been right, that common belief is an important component for effectively ensuring cooperation.  But as we more and more frequently consider what it \emph{means} for AI systems to cooperate with humans and how we can assess how cooperative an agent is, it seems important to recognize that standard common belief need not necessarily be a component.
} 

\section{On the Necessity and Sufficiency of Action-Stamped Common
Belief for Joint Behavior}\label{sec:correctness} 

We've argued in this paper that the prior work was incorrect in
asserting that common belief was necessary for joint behavior, and shown
by example that action-stamped common belief can suffice.
We now argue that an even stronger statement is true: there is a sense
in which action-stamped common belief is necessary and sufficient for
joint behavior.  We say ``in a sense'' here, because much depends on the
conception of joint behavior being considered.  So what we do in this
section is give a property that we would argue is one that we would
want to hold of joint behavior, and then show that action-stamped common belief is
necessary and sufficient for this property to hold.

What does it take to go from a collection of individual behaviors to a
joint behavior?  The following example may help illuminate some of the
relevant issues.
\begin{quote}
Jasper and Horace are both crooks, though neither is an evil
genius by any stretch of the imagination.  They both
independently decide to rob the Great Bank
of London 
on exactly the same day.
As it turns out, neither of them did a good job preparing, and
they each knew about only half of the bank's security systems,
and so made plans to bypass only that half.  By sheer dumb
luck, between them they know about all the bank's security
systems.  So when each bypasses the part that they know about
(at roughly the same time), the bank's security systems go down. 
They each make it in, steal a small fortune, and escape, none
the wiser as to the other's behavior or that their plan was
doomed to fail on its own. 
\end{quote}

Is Jasper and Horace robbing the bank an instance of joint behavior?  We
think not.  One critical component that distinguishes this from a
joint behavior is the beliefs of the agents.  Joint behaviors are
collective actions where people do their part because they believe
that everyone else will do their part as well.  Here, Jasper and
Horace have no inkling that the other will help disable the
system. 

We now want to capture these intuitions more formally.
 We start by adding another special Boolean variable $SHOULD\_ACT_{i,S}$ for
each agent $i$ and indexical group $S$, specifying the points in each
run where agent $i$ is 
supposed to act towards the plan of group $S$. 
We then add a special formula $\chi_S$ to the language:%
\footnote{As long as the set of agents is finite (which we implicitly
assume it is), we can express $\chi_S$ in a language that includes a
standard modal operator $\boxdot$, where $\boxdot \varphi$ is true at a
point $(r,n)$ iff $\varphi$ is true at all points $(r,n')$ in the run.
For ease of exposition, we do not introduce the richer modal logic here.}
\begin{itemize}
\item $(M,r,n) \vDash \chi_S$ if for all $n'$ and all agents
$i \in S(r,n')$, $(M,r,n') \vDash SHOULD\_ACT_{i,S} \rightarrow
ACTING_{i,S}$ 
\end{itemize}
The formula $\chi_S$ is thus true at a point $(r,n)$ if, at all points in
run $r$, each agent $i$ in the indexical group $S$ plays its part in
the group plan whenever it 
is supposed to.  


If we think of $ACTING_{i,S}$ as ``$i$ is taking part in the
joint behavior of the group $S$'', then the property JB$_S$ (for ``Joint Behavior, group $S$'') that we
now specify 
essentially says that to have truly joint behavior, each agent
in $S$ must believe when she acts that all of the members of the
(indexical) group
$S$ will
do what they're supposed to; if they don't all have that belief,
then
it's not really joint behavior.
Formally, JB$_S$ is a property of an indexical group $S$ in a model $M$: 

\smallskip
[JB$_S$:]	For all points $(r,n)$ and agents $i \in S(r,n)$, $ (M,r,n)
\vDash ACTING_{i,S} \rightarrow B_i \chi_S$.
\smallskip



Requiring JB$_S$ for joint behavior makes action-stamped common
belief of $\chi_S$ necessary for joint behavior.

\begin{theorem}\label{thm:necessity}
If JB$_S$ holds in  a model $M$,
then $(M,r,n) \vDash C^a_S \chi_S$ for all points $(r,n)$.
\end{theorem}

\begin{proof}
We begin by defining a notion of \emph{$a$-reachability}: A
point $(r',n')$ is $S$-$a$-reachable from $(r,n)$ in $k$ steps
if there exists a sequence $(r_0, n_0), \dots, (r_k, n_k)$ of
points such that $(r_0, n_0) = (r,n)$, $(r_k, n_k) = (r', n')$,
and for all $0 \leq l < k$, there exists a point
$(r_l, n'_l)$ and an agent $i \in S(r_l,n'_l)$ such that
$(M, r_l, n'_l) \vDash 
ACTING_{i,S}$ \sloppy{and $((r_l, n'_l),(r_{l+1}, n_{l+1})) \in
\B_i$.} 

By the semantics of $C^a_S$, $C^a_S \chi_S$ holds at $(r,n)$ iff
$\chi_S$ holds at every point $(r', n')$ that is
$S$-$a$-reachable from $(r,n)$ in $1$ or more steps.  Consider
any such point $(r', n')$.  Then, by the definition of reachability,
there exists some point $(r'', n'')$ and some agent $i \in
S(r'',n'')$ such that 
$(M, r'', n'') \vDash ACTING_{i,S}$ and $((r'', n''), (r',
n')) \in \B_i$.  Because $(M, r'', n'') \vDash
ACTING_{i,S}$, we get by JB$_S$ that
$(M, r'', n'') \vDash B_i \chi_S$.  Then by the semantics of
$B_i$ and the fact that $((r'', n''), (r', n')) \in \B_i$ we
get that $(M, r', n') \vDash \chi_S$.  But $(r',n')$ was an
arbitrary point $S$-$a$-reachable from $(r,n)$ in $1$ or more
steps, so $\chi_S$ holds at all such points, and we have that
$(M,r,n) \vDash C^{a}_S \chi_S$.  But $(r,n)$ was also
arbitrary, so $C^{a}_S \chi_S$ holds at all points. 
\end{proof}	

\commentout{
Consider an arbitrary $n$.  Our proof proceeds by induction to show that $M,r,n \vDash E^{a,k}_{G} \chi$ for all $k$.
For the base case, note simply that
property~\ref{property:ascb} holds at all points and so it
immediately follows that at all points where $GROUP\_ACTING_i$
holds for an agent $i$ it is also the case that $B_i \chi$, so
we have that $M,r,n \vDash E^{a,1}_{G} \chi$. 

For the inductive case, we'll assume that $M,r,n \vDash E^{a,k}_{G} \chi$, and show that $M,r,n \vDash E^{a,k+1}_{G} \chi$.  This is equivalent to assuming $\chi$ is true at every point that is $G$-$a$-reachable from $(r,n)$ in $k$ steps and showing that it is true at every point that is $G$-$a$-reachable from $(r,n)$ in $k+1$ steps.  Assume for the sake of contradiction that there is a point $G$-$a$-reachable from $(r,n)$ in $k+1$ steps where $\chi$ isn't true, that is, there is an agent $i$, a point $(r_{k+1}, n_{k+1})$, and a point $(r_k, n_k)$ that is $G$-$a$-reachable from $(r,n)$ in $k$ steps such that $(M, r_k, n'_k) \vDash GROUP\_ACTING_i$,  $((r_k, n'_k),(r_{k+1}, n_{k+1})) \in \B_i$, and $(M, r_{k+1}, n_{k+1}) \nvDash \chi$.  Choose an arbitrary such $(r_{k+1}, n_{k+1})$ and $(r_k, n'_k)$. But because property~\ref{property:ascb} holds at all points, $(M, r_k, n'_k) \vDash GROUP\_ACTING_i$ implies $(M, r_k, n'_k) \vDash B_i \chi$, which combined with the fact that $((r_k, n'_k),(r_{k+1}, n_{k+1})) \in \B_i$ implies that $(M, r_{k+1}, n_{k+1}) \vDash \chi$.  But this contradicts our assumption that $\chi$ isn't true at $(r_{k+1}, n_{k+1})$, and it was an arbitrary such point, so we get that $\chi$ is true at all points $G$-$a$-reachable from $(r,n)$ in $k+1$ steps, completing the proof.
} 


The converse to Theorem~\ref{thm:necessity} also holds; that is,
action-stamped common belief of $\chi_S$ suffices for JB$_S$ to
hold.  Put another way, action-stamped common belief is
exactly the ingredient that we need to meet the belief requirements 
of the property that we used to characterize
joint behavior. 

\begin{theorem}\label{thm:sufficient}
If $(M,r,n)
\vDash C^a_S \chi_S$ for all points $(r,n)$, then
JB$_S$ holds in $M$.
\end{theorem}

\begin{proof}
Consider an arbitrary point $(r,n)$ and agent $i \in S(r,n)$ such
that $(M,r,n) 
\vDash ACTING_{i,S}$.  By assumption, $(M,r,n) \vDash C^a_S
\chi_S$.  So, by the semantics of $C^a_S$, it follows that
$(M,r,n) \vDash E^a_S \chi_S$.  In turn, it follows from the
semantics of $E^a_S$ that $(M,r,n) \vDash B_i \chi_S$
(because $(M,r,n) \vDash ACTING_{i,S}$).  
But $r$, $n$, and $i$ were
arbitrary, so we have that $(M,r,n) \vDash ACTING_{i,S}
\rightarrow B_i \chi_S$ for all such points and agents.
Thus, JB$_S$ holds in $M$.
\end{proof}

The astute reader will have noticed that  the
proofs of   Theorem~\ref{thm:necessity} and \ref{thm:sufficient}
did not
depend in any way on $\chi_S$.  The formula $\chi_S$ in these
theorems can be replaced by an arbitrary
formula $\varphi$.  
In other words, if all the agents in $S$
believe $\varphi$ at the point when they act,
then $\varphi$ is
action-stamped common belief, and if $\varphi$ is action-stamped
common belief, then all agents in $S$ must believe $\varphi$ at
the point when they act.
Formally, the proofs of Theorem~\ref{thm:necessity} and
~\ref{thm:sufficient} also show the
following: 

\begin{theorem}\label{thm:necessity-general}
If $(M,r,n)
\vDash ACTING_{i,S} \rightarrow B_i \varphi$ for all points $(r,n)$
and agents $i \in S(r,n)$, then $(M,r,n) \vDash C^a_S
\varphi$ for all points $(r,n)$. 
\end{theorem}



\begin{theorem}\label{thm:sufficient-general}
If $(M,r,n)
\vDash C^a_S \varphi$ for all points $(r,n)$, then
$(M,r,n) \vDash ACTING_{i,S} \rightarrow B_i \varphi$ for
all points $(r,n)$ and agents $i \in S(r,n)$. 
\end{theorem}


\commentout{
\begin{proof}
Consider an arbitrary $r$, $n$, and $i \in S(r,n)$ such that $(M,r,n)
\vDash ACTING_{i,S}$.  By assumption $(M,r,n) \vDash C^a_S
\varphi$.  So by the semantics of $C^a_S$ it follows that
$(M,r,n) \vDash E^a_S \varphi$.  In turn, it follows from the
semantics of $E^a_S$ that $(M,r,n) \vDash B_i \varphi$ for all $i
\in S(r,n)$ (because
$(M,r,n) \vDash GROUP\_ACTING_i$).  But $r$, $n$, and $i$ were
arbitrary, so we have that $(M,r,n) \vDash GROUP\_ACTING_i
\rightarrow B_i \varphi$ for all such points and agents. 
\end{proof}
}


\commentout{
We now briefly sketch an even stronger sense in which action-stamped
common belief suffices for joint behavior.   The idea is that if, for
all point $(r,n)$ and 
agents $i$, $i$  act at $(r,n)$ if  (a) $i \in S(r,n)$,
(b) $i$ is supposed to act at that point (i.e., $SHOULD\_ACT_{i,s}$ holds),
and (c) $i$ believes $C^a_S \chi$, then the agents will be acing
jointly.  So, roughly speaking, if agents act when they have
action-stamped common belief they are acting jointly, they really wil
be acting jointly.  To make this precise, we need to use the idea of a
\emph{knowledge-based program} \cite{FHMV}.  A knowledge-based
program, as the name suggests, has explicit tests for knowledge (or
belief).  For example, the following is a knowledge-based program for
an agent $i$:
$$\mbox{{\bf if} $B_i B_j \varphi$ {\bf do} $a$}.$$
According to this program, if $i$ believes that $j$ believes $\varphi$,
then $i$ should perform action $a$.
A knowledge-based program for agent $i$ consists of a collection of
such statements, where each belief formula in a test starts with $B_i$ (so the
outermost modality involves $i$'s beliefs).

We can associate with each knowledge-based program $\Pg$ a set of models
that \emph{represent} $\Pg$.  The details are beyond the scope of this
paper (and can be found in \cite{FHMV}); we just give the high-level
idea here.  We need to assume that poins of a run have more structure;
specifically, they have the form $(s_1, \ldots, s_n)$, where $s_i$ is
the \emph{local state} of agent $i$.  Intuitively, agent $i$'s local
state encodes all the information that agent $i$ has access to.  A
\emph{protocol} for agent $i$ is a function from agent $i$'s local
states to actions; intuitively, a protocol for agent $i$ tells agent
$i$ what to do in each local state.  Given a program $\Pg_i$ for agent
$i$ and a model $M$, we can define a protocol $P_i$ for agent; in
local state $s_i$.  The truth of a formula of the form $B_i \varphi$
depends only on $i$'s local state at the point $(r,n)$.   Given a
local state $s_i$ for agent, we determine $P_i(s_i)$ by finding a
point $(r,n)$ such that $i$'s local state at that point is $s_i$ (if
such a point exists), and determining which of the tests in $\Pg_i$
are true at the point $(r,n)$.  Since the truth of a formula $B_i
\psi$ depends only on $i$'s local state, it doesn't matter which point
$(r,n)$ we choose.  According to $P_i(s_i)$, agent $i$ performs the
action at the first line of $\Pg_i$ where the corresponding test is
true at $(r,n)$.  If none of the tests are true, or if there is no
point $(r,n)$ where $i$'s local state is $s_i$, then $i$ performs some
default action.  
}

\section{Conclusion and Future Work}\label{sec:conclusion}

We have argued here that, contrary to what was suggested in earlier work,
common belief is not necessary for joint behavior.
We have presented a new notion, \emph{action-stamped} common
belief, and shown that it is, in a sense, necessary and sufficient for
joint behavior, and can be
achieved in scenarios where standard common belief cannot.  
This is important because modelling the conditions needed for joint
behavior correctly can enable 
cooperation in important scenarios, such as search and rescue, where
it might not otherwise be possible. 
%
\fullv{We chose to use the term \emph{joint behavior} in this paper because
it sounded to our ears like it most accurately captured the notion we
were considering; no doubt to some readers other terms will sound like
a better fit.}
As we showed in Section~\ref{sec:correctness}, action-stamped common
belief characterizes scenarios where individuals do their part only if
they believe others will do the 
\fullv{same, whatever terminology we use.}
\shortv{same.}

We suspect that, for some readers, the idea that action-stamped common
belief is sufficient for joint behavior will seem obvious.
In a certain sense, we agree; in retrospect, it
\emph{does} feel like the obviously correct notion for joint
behaviors.  
That said, while action-stamped common belief seems
quite natural, it does not seem to have been studied in any prior literature.


With that in mind, it is worth  briefly discussing the connection between the ideas in
this paper and some of the prior work that has been done.  First, 
note that action-stamped common belief is in some ways the
natural variant of common belief for extensive-form games.  Because
an agent $i$'s information sets are usually specified only at nodes at
which agent $i$ moves, it is possible to reason about agent $i$'s beliefs
only at points where agent $i$ acts.  This makes it all the more surprising
that action-stamped common belief has not been formalized and studied
in its own right; in some sense, it captures what epistemic game
theorists have been implicitly considering in the case of extensive-form games. 

In this paper, we have considered the types of \emph{beliefs} necessary for
joint behavior, but that may not be the only factor involved (nor do
we claim it is; we  are just focused in this paper on the belief
component).  For example, in much of the literature, \emph{intent} is
taken to play an important role in various cooperative behaviors.   
Dunin-Keplicz and Verbrugge~\cite{DKV04} proposed a three-part notion
of ``collective commitment'', with the levels of belief (e.g., no one
believes, everyone believes, it is common belief) held at each of the
three parts leading to various types of collective commitment. Their
work is in some ways orthogonal to ours; it can be thought of
as considering various types of cooperative behaviors that can occur,
while ours just focuses on one particularly strict form, joint
behavior.  One way of interpreting our work in the context of theirs
is as saying that the top level of belief to consider for cooperation
should in fact be action-stamped common belief. 

Blomberg~\cite{Blomberg16} gives an insightful argument that
common belief (and variants thereof) of intentions is not necessary
for a joint intentional act.
Roughly speaking,
an agent may (incorrectly) believe that other agents do not share
his intent, as long as what he believes they intend
would still lead them to act in the manner conducive to his goals. 
We find his counterexample and
arguments compelling.
But this is perfectly consistent with our results.
Theorems \ref{thm:necessity} and \ref{thm:sufficient} show that
action-stamped common belief (or in the case of simultaneous acts,
standard common belief) that agents will do the necessary acts
is required for joint behaviors.
We place no requirements on what agents have to believe about other
agents' intentions.
Put a different way, 
Blomberg makes a compelling case that, when
characterizing joint behavior, it is a mistake to instantiate the $\varphi$ in
Theorems \ref{thm:necessity-general} and \ref{thm:sufficient-general}
with formulas about shared intents.
That is to say, it is not a necessary property of
cooperative behavior that agents act only if they are sure 
others are acting with the same purpose.

Ludwig~\cite{Ludwig07,Ludwig16} also presents an argument that
common belief is not necessary for joint (intentional) action.
Putting aside the question of whether his argument is correct,
it is not relevant to our considerations here as it relies on a much
broader notion of cooperative behavior than the joint behaviors we
consider in this 
paper (though he calls it ``joint action'').  That is, we certainly
agree with his conclusion that for \emph{some} types of cooperative
behavior agents don't need to be sure that others will do their part and so
don't need common belief, but for the types of cooperative behaviors
we consider in this paper we have shown in Theorem
\ref{thm:necessity-general} that (action-stamped) common belief is in
fact necessary. 

Lastly, Roy and Schwenkenbecher~\cite{RS21} consider a novel
notion of belief that they call ``pooled knowledge'', which is related to
distributed knowledge, and argue that it is both weaker than common
knowledge and sufficient for shared intentions.  The basic idea behind
the argument is that if agents are rational, then pooled knowledge
would induce some agent to act as a coordinator to guide the behavior
of the group.  It's certainly an interesting proposal, and one that
deserves further study.  From the perspective of this paper, it would
be interesting to try to formally analyze under what conditions pooled
knowledge/belief would lead to action-stamped common
knowledge/belief. 


The present work suggests two areas that are ripe for future work.   
The first is to more fully explore the logical aspects of
action-stamped common belief.  Can a sound and complete axiomatization
be provided?  What is the complexity
of the model-checking and validity problems for a language involving
action-stamped common knowledge?
How can we practically engineer systems that rely
on action-stamped common belief?
The second area we think worth exploring is
that of understanding better what levels of group knowledge are
required for other aspects of joint behavior and other types of
cooperation.  We focused on one aspect, revealing a
nuanced but important error in earlier thinking.  We think that there
may well be other aspects of cooperation that are worth digging into
in this fine-grained way. 
\fullv{Given the importance of cooperative AI, we hope that others will join
us in exploring these questions. }

\section*{Acknowledgments}
Meir Friedenberg and Joe Halpern were supported in part by 
ARO grant W911NF-22-1-0061 and MURI grant W911NF-19-1-0217.  Halpern
was also supported in part by a AFOSR grant FA9550-12-1-0040 and a
grant from the Algorand Centres of Excellence program managed by
Algorand Foundation.  Any opinions, findings, and conclusions or
recommendations expressed in this material are those of the author(s)
and do not necessarily reflect the views of the funders.

\bibliographystyle{eptcs}
\bibliography{z,joe}

\begin{thebibliography}{10}
\providecommand{\bibitemdeclare}[2]{}
\providecommand{\surnamestart}{}
\providecommand{\surnameend}{}
\providecommand{\urlprefix}{Available at }
\providecommand{\url}[1]{\texttt{#1}}
\providecommand{\href}[2]{\texttt{#2}}
\providecommand{\urlalt}[2]{\href{#1}{#2}}
\providecommand{\doi}[1]{doi:\urlalt{https://doi.org/#1}{#1}}
\providecommand{\eprint}[1]{arXiv:\urlalt{https://arxiv.org/abs/#1}{#1}}
\providecommand{\bibinfo}[2]{#2}

\bibitemdeclare{article}{Blomberg16}
\bibitem{Blomberg16}
\bibinfo{author}{O.~\surnamestart Blomberg\surnameend} (\bibinfo{year}{2016}):
  \emph{\bibinfo{title}{Common knowledge and reductionism about shared
  agency}}.
\newblock {\slshape \bibinfo{journal}{Austalasian Journal of Philosophy}}
  \bibinfo{volume}{94}(\bibinfo{number}{2}), pp. \bibinfo{pages}{315--326},
  \doi{10.1080/00048402.2015.1055581}.

\bibitemdeclare{article}{bratman92}
\bibitem{bratman92}
\bibinfo{author}{M.~E. \surnamestart Bratman\surnameend}
  (\bibinfo{year}{1992}): \emph{\bibinfo{title}{Shared cooperative activity}}.
\newblock {\slshape \bibinfo{journal}{The Philosophical Review}}
  \bibinfo{volume}{101}(\bibinfo{number}{2}), pp. \bibinfo{pages}{327--341},
  \doi{10.2307/2185537}.

\bibitemdeclare{article}{CM03}
\bibitem{CM03}
\bibinfo{author}{J.~\surnamestart Casper\surnameend} \& \bibinfo{author}{R.~R.
  \surnamestart Murphy\surnameend} (\bibinfo{year}{2003}):
  \emph{\bibinfo{title}{Human-robot interactions during the robot-assisted
  urban search and rescue response at the world trade center}}.
\newblock {\slshape \bibinfo{journal}{IEEE Transactions on Systems, Man, and
  Cybernetics, Part B (Cybernetics)}}
  \bibinfo{volume}{33}(\bibinfo{number}{3}), pp. \bibinfo{pages}{367--385},
  \doi{10.1109/TSMCB.2003.811794}.

\bibitemdeclare{article}{CL91}
\bibitem{CL91}
\bibinfo{author}{P.~R. \surnamestart Cohen\surnameend} \&
  \bibinfo{author}{H.~J. \surnamestart Levesque\surnameend}
  (\bibinfo{year}{1991}): \emph{\bibinfo{title}{Teamwork}}.
\newblock {\slshape \bibinfo{journal}{Nous}}
  \bibinfo{volume}{25}(\bibinfo{number}{4}), pp. \bibinfo{pages}{487--512},
  \doi{10.2307/2216075}.

\bibitemdeclare{misc}{CAIFcall}
\bibitem{CAIFcall}
\bibinfo{author}{\surnamestart {Cooperative AI Foundation (CAIF)}\surnameend}
  (\bibinfo{year}{2022}): \emph{\bibinfo{title}{Evaluation for Cooperative
  {AI}: Call for Proposals}}.
\newblock
  \bibinfo{note}{{h}ttps://www.cooperativeai.com/calls-for-proposals/evaluation-for-cooperative-ai;
  accessed October. 25, 2022}.

\bibitemdeclare{misc}{DBHHLG21}
\bibitem{DBHHLG21}
\bibinfo{author}{A.~\surnamestart Dafoe\surnameend},
  \bibinfo{author}{Y.~\surnamestart Bachrach\surnameend},
  \bibinfo{author}{G.~\surnamestart Hadfield\surnameend},
  \bibinfo{author}{E.~\surnamestart Horvitz\surnameend},
  \bibinfo{author}{K.~\surnamestart Larson\surnameend} \&
  \bibinfo{author}{T.~\surnamestart Graepel\surnameend} (\bibinfo{year}{2021}):
  \emph{\bibinfo{title}{Cooperative {AI}: machines must learn to find common
  ground}}.

\bibitemdeclare{unpublished}{Dafoe20}
\bibitem{Dafoe20}
\bibinfo{author}{A.~\surnamestart Dafoe\surnameend},
  \bibinfo{author}{E.~\surnamestart Hughes\surnameend},
  \bibinfo{author}{Y.~\surnamestart Bachrach\surnameend},
  \bibinfo{author}{T.~\surnamestart Collins\surnameend}, \bibinfo{author}{K.~R.
  \surnamestart McKee\surnameend}, \bibinfo{author}{J.~Z. \surnamestart
  Leibo\surnameend}, \bibinfo{author}{K.~\surnamestart Larson\surnameend} \&
  \bibinfo{author}{T.~\surnamestart Graepel\surnameend} (\bibinfo{year}{2020}):
  \emph{\bibinfo{title}{Open problems in cooperative {AI}}}.
\newblock \bibinfo{note}{Available at https://arxiv.org/pdf/2012.08630.pdf}.

\bibitemdeclare{article}{DKV04}
\bibitem{DKV04}
\bibinfo{author}{B.~\surnamestart Dunin-Kęplicz\surnameend} \&
  \bibinfo{author}{R.~\surnamestart Verbrugge\surnameend}
  (\bibinfo{year}{2004}): \emph{\bibinfo{title}{A tuning machine for
  cooperative problem solving}}.
\newblock {\slshape \bibinfo{journal}{Fundamenta Informaticae}}
  \bibinfo{volume}{63}(\bibinfo{number}{2--3}), pp. \bibinfo{pages}{283--307}.

\bibitemdeclare{book}{FHMV}
\bibitem{FHMV}
\bibinfo{author}{R.~\surnamestart Fagin\surnameend}, \bibinfo{author}{J.~Y.
  \surnamestart Halpern\surnameend}, \bibinfo{author}{Y.~\surnamestart
  Moses\surnameend} \& \bibinfo{author}{M.~Y. \surnamestart Vardi\surnameend}
  (\bibinfo{year}{1995}): \emph{\bibinfo{title}{Reasoning About Knowledge}}.
\newblock \bibinfo{publisher}{MIT Press}, \bibinfo{address}{Cambridge, MA}.
\newblock \bibinfo{note}{A slightly revised paperback version was published in
  2003.}

\bibitemdeclare{article}{GK96}
\bibitem{GK96}
\bibinfo{author}{B.~J. \surnamestart Grosz\surnameend} \&
  \bibinfo{author}{S.~\surnamestart Kraus\surnameend} (\bibinfo{year}{1996}):
  \emph{\bibinfo{title}{Collaborative plans for complex group action}}.
\newblock {\slshape \bibinfo{journal}{Artificial Intelligence}}
  \bibinfo{volume}{86}(\bibinfo{number}{2}), pp. \bibinfo{pages}{269--357},
  \doi{10.1016/0004-3702(95)00103-4}.

\bibitemdeclare{incollection}{GrS90}
\bibitem{GrS90}
\bibinfo{author}{B.~J. \surnamestart Grosz\surnameend} \&
  \bibinfo{author}{C.~L. \surnamestart Sidner\surnameend}
  (\bibinfo{year}{1990}): \emph{\bibinfo{title}{Plans for discourse}}.
\newblock In \bibinfo{editor}{P.~R. \surnamestart Cohen\surnameend},
  \bibinfo{editor}{J.~L. \surnamestart Morgan\surnameend} \&
  \bibinfo{editor}{M.~E. \surnamestart Pollack\surnameend}, editors: {\slshape
  \bibinfo{booktitle}{Intentions in Communication}},
  chapter~\bibinfo{chapter}{20}, \bibinfo{publisher}{MIT Press}.

\bibitemdeclare{article}{HM90}
\bibitem{HM90}
\bibinfo{author}{J.~Y. \surnamestart Halpern\surnameend} \&
  \bibinfo{author}{Y.~\surnamestart Moses\surnameend} (\bibinfo{year}{1990}):
  \emph{\bibinfo{title}{Knowledge and common knowledge in a distributed
  environment}}.
\newblock {\slshape \bibinfo{journal}{Journal of the ACM}}
  \bibinfo{volume}{37}(\bibinfo{number}{3}), pp. \bibinfo{pages}{549--587},
  \doi{10.1145/79147.79161}.

\bibitemdeclare{article}{KT01}
\bibitem{KT01}
\bibinfo{author}{H.~\surnamestart Kitano\surnameend} \&
  \bibinfo{author}{S.~\surnamestart Tadokoro\surnameend}
  (\bibinfo{year}{2001}): \emph{\bibinfo{title}{Robocup rescue: A grand
  challenge for multiagent and intelligent systems}}.
\newblock {\slshape \bibinfo{journal}{AI magazine}}
  \bibinfo{volume}{22}(\bibinfo{number}{1}), pp. \bibinfo{pages}{39--39},
  \doi{10.1609/aimag.v37i1.2642}.

\bibitemdeclare{conference}{KTNMTSS99}
\bibitem{KTNMTSS99}
\bibinfo{author}{H.~\surnamestart Kitano\surnameend},
  \bibinfo{author}{S.~\surnamestart Tadokoro\surnameend},
  \bibinfo{author}{I.~\surnamestart Noda\surnameend},
  \bibinfo{author}{H.~\surnamestart Matsubara\surnameend},
  \bibinfo{author}{T.~\surnamestart Takahashi\surnameend},
  \bibinfo{author}{A.~\surnamestart Shinjou\surnameend} \&
  \bibinfo{author}{S.~\surnamestart Shimada\surnameend} (\bibinfo{year}{1999}):
  \emph{\bibinfo{title}{Robocup rescue: Search and rescue in large-scale
  disasters as a domain for autonomous agents research}}.
\newblock In: {\slshape \bibinfo{booktitle}{1999 IEEE International Conference
  on Systems, Man, and Cybernetics}}, \bibinfo{volume}{6}, pp.
  \bibinfo{pages}{739--743}.

\bibitemdeclare{inproceedings}{LVC90}
\bibitem{LVC90}
\bibinfo{author}{H.~J. \surnamestart Levesque\surnameend},
  \bibinfo{author}{P.~R. \surnamestart Cohen\surnameend} \&
  \bibinfo{author}{J.~H.~T. \surnamestart Nunes\surnameend}
  (\bibinfo{year}{1990}): \emph{\bibinfo{title}{On acting together}}.
\newblock In: {\slshape \bibinfo{booktitle}{Proc.~Eighth National Conference on
  Artificial Intelligence (AAAI '90)}}, pp. \bibinfo{pages}{94--99}.

\bibitemdeclare{article}{LN13}
\bibitem{LN13}
\bibinfo{author}{Y.~\surnamestart Liu\surnameend} \&
  \bibinfo{author}{G.~\surnamestart Nejat\surnameend} (\bibinfo{year}{2013}):
  \emph{\bibinfo{title}{Robotic urban search and rescue: A survey from the
  control perspective}}.
\newblock {\slshape \bibinfo{journal}{Journal of Intelligent \& Robotic
  Systems}} \bibinfo{volume}{72}(\bibinfo{number}{2}), pp.
  \bibinfo{pages}{147--165}, \doi{10.1007/s10846-013-9822-x}.

\bibitemdeclare{article}{Ludwig07}
\bibitem{Ludwig07}
\bibinfo{author}{K.~\surnamestart Ludwig\surnameend} (\bibinfo{year}{2007}):
  \emph{\bibinfo{title}{Collective intentional behavior from the standpoint of
  semantics}}.
\newblock {\slshape \bibinfo{journal}{No{\^u}s}}
  \bibinfo{volume}{41}(\bibinfo{number}{3}), pp. \bibinfo{pages}{355--393},
  \doi{10.1111/j.1468-0068.2007.00652.x}.

\bibitemdeclare{book}{Ludwig16}
\bibitem{Ludwig16}
\bibinfo{author}{K.~\surnamestart Ludwig\surnameend} (\bibinfo{year}{2016}):
  \emph{\bibinfo{title}{From Individual to Plural Agency: Collective Action}}.
\newblock \bibinfo{publisher}{Oxford University Press},
  \doi{10.1093/acprof:oso/9780198755623.001.0001}.

\bibitemdeclare{article}{QTP20}
\bibitem{QTP20}
\bibinfo{author}{J.~P. \surnamestart Queralta\surnameend},
  \bibinfo{author}{J.~\surnamestart Taipalmaa\surnameend},
  \bibinfo{author}{B.~C. \surnamestart Pullinen\surnameend},
  \bibinfo{author}{V.~K. \surnamestart Sarker\surnameend},
  \bibinfo{author}{T.~N. \surnamestart Gia\surnameend},
  \bibinfo{author}{H.~\surnamestart Tenhunen\surnameend},
  \bibinfo{author}{M.~\surnamestart Gabbouj\surnameend},
  \bibinfo{author}{J.~\surnamestart Raitoharju\surnameend} \&
  \bibinfo{author}{T.~\surnamestart Westerlund\surnameend}
  (\bibinfo{year}{2020}): \emph{\bibinfo{title}{Collaborative multi-robot
  search and rescue: Planning, coordination, perception, and active vision}}.
\newblock {\slshape \bibinfo{journal}{IEEE Access}} \bibinfo{volume}{8}, pp.
  \bibinfo{pages}{191617--191643}, \doi{10.1109/ACCESS.2020.3030190}.

\bibitemdeclare{inproceedings}{RS97}
\bibitem{RS97}
\bibinfo{author}{C.~\surnamestart Rich\surnameend} \& \bibinfo{author}{C.~L.
  \surnamestart Sidner\surnameend} (\bibinfo{year}{1997}):
  \emph{\bibinfo{title}{{COLLAGEN:} When agents collaborate with people}}.
\newblock In: {\slshape \bibinfo{booktitle}{Proceedings of the First
  International Conference on Autonomous Agents}}, pp.
  \bibinfo{pages}{284--291}, \doi{10.1145/267658.267730}.

\bibitemdeclare{article}{RS21}
\bibitem{RS21}
\bibinfo{author}{O.~\surnamestart Roy\surnameend} \&
  \bibinfo{author}{A.~\surnamestart Schwenkenbecher\surnameend}
  (\bibinfo{year}{2021}): \emph{\bibinfo{title}{Shared intentions, loose
  groups, and pooled knowledge}}.
\newblock {\slshape \bibinfo{journal}{Synthese}}
  \bibinfo{volume}{198}(\bibinfo{number}{5}), pp. \bibinfo{pages}{4523--4541},
  \doi{10.1007/s11229-019-02355-x}.

\bibitemdeclare{conference}{SYDY04}
\bibitem{SYDY04}
\bibinfo{author}{J.~\surnamestart Scholtz\surnameend},
  \bibinfo{author}{J.~\surnamestart Young\surnameend}, \bibinfo{author}{J.~L.
  \surnamestart Drury\surnameend} \& \bibinfo{author}{H.~A. \surnamestart
  Yanco\surnameend} (\bibinfo{year}{2004}): \emph{\bibinfo{title}{Evaluation of
  human-robot interaction awareness in search and rescue}}.
\newblock In: {\slshape \bibinfo{booktitle}{Proc.~IEEE International Conference
  on Robotics and Automation, 2004 (ICRA'04)}}, \bibinfo{volume}{3}, pp.
  \bibinfo{pages}{2327--2332}, \doi{10.1109/ROBOT.2004.1307409}.

\bibitemdeclare{article}{tambe97}
\bibitem{tambe97}
\bibinfo{author}{M.~\surnamestart Tambe\surnameend} (\bibinfo{year}{1997}):
  \emph{\bibinfo{title}{Towards flexible teamwork}}.
\newblock {\slshape \bibinfo{journal}{Journal of A.I. Research}}
  \bibinfo{volume}{7}, pp. \bibinfo{pages}{83--124}.

\bibitemdeclare{article}{Tuomela05}
\bibitem{Tuomela05}
\bibinfo{author}{R.~\surnamestart Tuomela\surnameend} (\bibinfo{year}{2005}):
  \emph{\bibinfo{title}{We-intentions revisited}}.
\newblock {\slshape \bibinfo{journal}{Philosophical Studies}}
  \bibinfo{volume}{125}(\bibinfo{number}{3}), pp. \bibinfo{pages}{327--369},
  \doi{10.1007/s11098-005-7781-1}.

\end{thebibliography}
\end{document}